\documentclass[runningheads,a4paper]{llncs}
\usepackage{amssymb}
\usepackage{graphicx}
\usepackage{algorithm,algpseudocode}
\usepackage{amsmath}
\usepackage{amsfonts}
\usepackage{stmaryrd}
\usepackage{color}

\usepackage[font=small,labelfont=bf]{caption}
\usepackage{url}

\begin{document}

\mainmatter  % start of an individual contribution

\title{Differentially Private Neighborhood-based Recommender Systems}

\author{Jun Wang$^1$
\and Qiang Tang$^2$ }
\institute{
%$^1$Interdisciplinary Center for Security, Reliability and Trust (SnT)\\
$^1$University of Luxembourg\\
$^2$Luxembourg Institute of Science and Technology\\
\email{jun.wang@uni.lu}; 
\email{tonyrhul@gmail.com}
}


%\toctitle{Lecture Notes in Computer Science}
%\tocauthor{Authors' Instructions}
\maketitle

\begin{abstract}
Privacy issues of recommender systems have become a hot topic for the society as such systems are appearing in every corner of our life. In contrast to the fact that many secure multi-party computation protocols have been proposed to prevent information leakage in the process of recommendation computation, very little has been done to restrict the information leakage from the recommendation results. In this paper, we apply the differential privacy concept to neighborhood-based recommendation methods (NBMs) under a probabilistic framework. We first present a solution, by directly calibrating Laplace noise into the training process, to differential-privately find the maximum a posteriori parameters \emph{similarity}. Then we connect differential privacy to NBMs by exploiting a recent observation that sampling from the scaled posterior distribution of a Bayesian model results in provably differentially private systems. Our experiments show that both solutions allow promising accuracy with a modest privacy budget, and the second solution yields better accuracy if the sampling asymptotically converges. We also compare our solutions to the recent differentially private matrix factorization (MF) recommender systems, and show that our solutions achieve better accuracy when the privacy budget is reasonably small. This is an interesting result because MF systems often offer better accuracy when differential privacy is not applied.
\end{abstract}

\section{Introduction}
Recommender systems, particularly collaborative filtering (CF) systems, have been widely deployed due to the success of E-commerce \cite{su2009survey}. There are two dominant approaches in CF. One is matrix factorization (MF) \cite{koren2009matrix} which models the user preference matrix as a product of two low-rank user and item feature matrices, and the other is neighborhood-based method (NBM) which leverages the \emph{similarity} between items or users to estimate user preferences \cite{desrosiers2011comprehensive}. Generally, MF is more accurate than NBM \cite{su2009survey}, while NBM has an irreplaceable advantage that it naturally explains the recommendation results. In addition, recent research shows that MF falls short in session-based recommendation while NBMs allow promising accuracy \cite{hidasi2015session}. Therefore, NBM is still an interesting research topic for the community.

In reality, industrial CF recommender and ranking systems often adopt a client-server model, in which a single server (or, server cluster) holds databases and serves a large number of users. CF exploits the fact that similar users are likely to prefer similar products, unfortunately this property facilitates effective user de-anonymization and history information recovery through the recommendation results \cite{calandrino2011you,narayanan2008robust}. To this end, NBM is more fragile (e.g. \cite{calandrino2011you,mobasher2007toward}), since it is essentially a simple linear combination of user history information which is weighted by the normalized \emph{similarity} between users or items. In this paper, we aim at preventing information leakage from the recommendation results, for the NBM systems. Note that a related research topic is to avoid the server from accessing the users' plaintext inputs, and many solutions exist for this (e.g. \cite{nikolaenko2013privacy,tang2015privacy}). Combining them with our solution will result in a comprehensive solution, which prevent information leakage from both the computation process and final recommendation results. We skip the details here.

Differential privacy \cite{dwork2014algorithmic} provides rigorous privacy protection for user information in statistical databases. Intuitively, it offers a participant the possibility to deny his participation in a computation. Some works, such as \cite{mcsherry2009differentially,zhu2014effective}, have been proposed for some specific NBMs, which adopt correlations or artificially defined metrics as \emph{similarity} \cite{desrosiers2011comprehensive} and are less appealing from the perspective of accuracy. It remains as an open issue to apply the differential privacy concept to more sophisticated NBM models, which automatically learn \emph{similarity} from training data (e.g. \cite{rendle2009bpr,toscher2008improved,PNBM}). Particularly, probabilistic NBM \cite{PNBM} models the dependencies among observations (ratings) which leads user preference estimation to a penalized risk minimization problem to search optimal unobserved factors (In our context, the unobserved factor is \emph{similarity}). It has been shown that the instantiation in \cite{PNBM} outperforms most other NBM systems and even the MF or probabilistic MF systems in many settings.

\subsection{Our Contribution}

Due to its accuracy advantages, we focus on the probabilistic NBM systems in our study. Inspired by \cite{berlioz2015applying,liu2015fast}, we propose two methods to instantiate differentially private solutions.

First, we calibrate noise into the training process (i.e. SGD) to differential-privately find the maximum a posteriori \emph{similarity}. This instantiation achieves differential privacy for each rating value. Second, we link the differential privacy concept to probabilistic NBM, by sampling from scaled posterior distribution. For the sake of efficiency, we employ a recent MCMC method, namely Stochastic Gradient Langevin Dynamics (SGLD) \cite{welling2011bayesian}, as the sampler. In order to use SGLD, we derive an unbiased estimator of \emph{similarity} gradient from a mini-batch. This instantiation achieves differential privacy for every user profile (rating vector).

To evaluate our solutions, we carry out experiments to compare our solutions to the state-of-the-art differentially private MFs, and also to compare our solutions between themselves. Our results show that differentially private MFs are more accurate when privacy loss is large (extremely, in a non-private case), but differentially private NBMs are better when privacy loss is set in a more reasonable range. Even with the added noises, both our solutions consistently outperform non-private traditional NBMs in accuracy. Despite the complexity concern, our solution with posterior sampling (i.e. SGLD) outperforms the other from the accuracy perspective.

\subsection{Organization}

The rest of this paper is organized as follows. In Section \ref{sec:pre}, we recap the preliminary knowledge. In Section \ref{sec:pgln} and \ref{sec:dpps}, we present our two differentially private NBM solutions respectively. In Section \ref{sec:exp}, we present our experiment results. In Section \ref{sec:related}, we present the related work. In Section \ref{sec:con}, we conclude the paper.

\section{Preliminary}
\label{sec:pre}

Generally, NBMs can be divided into user-user approach (relies on \emph{similarity} between users) and item-item approach (relies on \emph{similarity} between items) \cite{desrosiers2011comprehensive}. Probabilistic NBM can be regarded as a generic methodology, to be employed by any other specific NBM system. Commonly, the item-item approach is more accurate and robust than the user-user approach \cite{desrosiers2011comprehensive,mobasher2007toward}. In this paper, we take the item-item approach as an instance to introduce the probabilistic NBM concept from \cite{PNBM}. We also review the concept of differential privacy.

\label{sec:pnbm-1}
\begin{table}[h!]
\centering
\caption{Notation}
\label{tab:notation}
\begin{tabular}{|l|l|}
\hline
$r_{ui}$ & the rating that user $u$ gave item $i$                                                                  \\ \hline
$s_{ij}$ & the similarity between item $i$ and $j$                                                                  \\ \hline
$R \in \mathbb{R}^{N\times M}$        & rating matrix                                                                  \\ \hline
$R^{>0} \subset R$        & all the observed ratings or training data \\ \hline
$S\in \mathbb{R}^{M\times M}$         & item similarity matrix \\ \hline
$S_i \in \mathbb{R}^{1\times M}$      & similarity vector of item $i$                                                  \\ \hline
$R_u^{-} \in \mathbb{R}^{M \times 1}$ & $u$'s rating vector without the item being modeled             \\ \hline
$\alpha_S,\alpha_R$                   & hyperparameters of $S_i$ and $r_{ui}$ respectively                     \\ \hline
$f(S_i,R_u^{-})$                      & any NBM which takes as input the $S_i$ and $R_u^{-}$         \\ \hline
$p(*)$                                & prior distribution of $*$                                                      \\ \hline
$p(S_i|\alpha_S)$                     & likelihood function of $S_i$ conditioned on $\alpha_S$                         \\ \hline
$p(r_{ui}|f(*),\alpha_R)$   & likelihood function of $r_{ui}$ \\ \hline
\end{tabular}
\end{table}

\begin{figure}[h!]
\centering
\includegraphics[height=1.6in, width=2.4in]{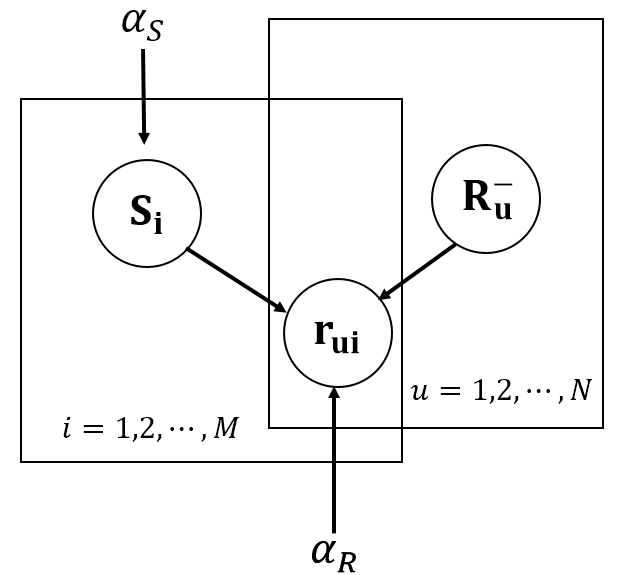}
\caption{Graphical model of PNBM}
\label{fig:gpnbm}
\end{figure}

\subsection{Review Probabilistic NBM}

Suppose we have a dataset with $N$ users and $M$ items. Probabilistic NBM \cite{PNBM} assumes the observed ratings $R^{>0}$ conditioned on historical ratings with Gaussian noise, see Fig. \ref{fig:gpnbm}. Some notation is summarized in Table \ref{tab:notation}. The likelihood function of observations $R^{>0}$ and prior of \emph{similarity} $S$ are written as
\begin{small}
\begin{equation}
\label{eq:cdd}
\begin{split}
p(R^{>0}|S, R^-, \alpha_R) = \prod_{i=1}^{M}\prod_{u=1}^{N}[\mathcal{N}(r_{ui}|f(S_i,R_u^{-}), \alpha_{R}^{-1}) ]^{I_{ui}};\ \ \ \ p(S|\alpha_{S})=\prod_{i=1}^{M} \mathcal{N}(S_{i}|0, \alpha_{S}^{-1}\mathbf{I})
\end{split}
\end{equation}
\end{small}
where $\mathcal{N}(x|\mu, \alpha^{-1})$ denotes the Gaussian distribution with mean $\mu$ and precision $\alpha$. $R^-$ indicates that if item $i$ is being modeled then it is excluded from the training data $R^{>0}$. $f(S_i,R_u^{-})$ denotes any NBM which takes as inputs the $S_i$ and $R_u^{-}$. In the following, we instantiate it to be a typical NBM \cite{desrosiers2011comprehensive}:
\begin{equation}\label{myeq:tnbm}
 \hat{r}_{ui} \leftarrow f(S_i,R_u^{-}) = \bar{r}_{i} + \frac{\sum_{j\in \mathcal{I} \backslash \{i\}}s_{ij}(r_{uj}-\bar{r}_j)I_{uj}}{\sum_{j \in \mathcal{I} \backslash \{i\}}|s_{ij}|I_{uj}}=\frac{S_iR_u^{-}}{|S_i|I_u^{-}} \quad
\end{equation}
$\hat{r}_{ui}$ denotes the estimation of user $u$'s preference on item $i$, $\bar{r}_i$ is item $i$'s mean rating value, $I_{uj}$ is the rating indicator  $I_{uj}=1$ if user $u$ rated item $j$, otherwise, $I_{uj}=0$. Similar with $R_u^- $, $I_u^-$ denotes user $u$'s indicator vector but set $I_{ui}=0$ if $i$ is the item being estimated. For the ease of notation, we will omit the term $\bar{r}_i$ and present Equation (\ref{myeq:tnbm}) in a vectorization form in favor of a slightly more succinct notation.

The log of the posterior distribution over the \emph{similarity} is
\begin{gather}
\label{myeq:logmpnbm2}
\begin{aligned}
- \log & p(S|R^{>0}, \alpha_S,\alpha_R) = -\log p(R^{>0}|S,R^-,\alpha_R)p(S|\alpha_S) =\\
& \frac{\alpha_R}{2} \sum_{i=1}^M\sum_{u=1}^N(r_{ui}-\frac{S_iR_u^-}{|S_i|I_u^-})^2+\frac{\alpha_s}{2}\sum_{i=1}^M(||S_i||_2)
+ M^2 \log \frac{\alpha_s}{\sqrt{2\pi}} + \log \frac{\alpha_R}{\sqrt{2\pi}} \sum_{i=1}^M\sum_{u=1}^NI_{ui}\end{aligned}\raisetag{3\baselineskip}
\end{gather}

%Essentially, the joint distribution $p(R^{-}, R, S_i, \alpha_S, \alpha_R)$ is the posterior distribution of \emph{similarity} $p(S|R^{>0},\alpha_S,\alpha_R)$.
Thanks to the simplicity of the log-posterior distribution (i.e. $\sum_{i=1}^M\sum_{u=1}^N(r_{ui}-\frac{S_iR_u^-}{|S_i|I_u^-})^2+\sum_{i=1}^M(||S_i||_2)$, where we omit the constant terms in Equation (\ref{myeq:logmpnbm2})). We can have two approaches to solve this risk minimization problem.

\begin{itemize}
\item \emph{Stochastic Gradient Descent (SGD).} In this approach,  $\log p(S|R^{>0}, \alpha_S,\alpha_R)$ is treated as an error function. SGD can be adopted to minimize the error function. In each SGD iteration we update the gradient of \emph{similarity} ($-\frac{\partial \log p(S|R^{>0}, \alpha_S,\alpha_R)}{\partial S_{ij}}$) with a set of randomly chosen ratings $\Phi$ by
%$S_{ij} \leftarrow S_{ij}-\eta ( \sum_{ (u , j) \in \Phi} (\hat{r}_{ui}-r_{ui}) \frac{\partial \hat{r}_{ui} }{\partial S_{ij}} + \lambda S_{ij})$.
\begin{equation}\label{myeq:tgradient}
S_{ij} \leftarrow S_{ij}-\eta ( \sum_{ (u , j) \in \Phi} (\hat{r}_{ui}-r_{ui}) \frac{\partial \hat{r}_{ui} }{\partial S_{ij}} + \lambda S_{ij})
\end{equation}
where $\eta$ is the learning rate, $\lambda= \frac{\alpha_S}{\alpha_R}$ is the regular parameter, the set $\Phi$ may contain $n \in [1, N]$ users.
In Section \ref{sec:pgln}, we will introduce how to build the differentially private SGD to train probabilistic NBM.

\item \emph{Monte Carlo Markov Chain (MCMC).} We estimate the predictive distribution of an unknown rating by a Monte Carlo approximation. In Section \ref{sec:dpps}, we will connect differential privacy to samples from the posterior $p(S|R^{>0}, \alpha_S,\alpha_R)$, via Stochastic Gradient Langevin Dynamics (SGLD) \cite{welling2011bayesian}.
\end{itemize}

\subsection{Differential Privacy}
\label{sec:dp}

Differential privacy \cite{dwork2014algorithmic}, which is a dominate security definition against inference attacks, aims to rigorously protect sensitive data in statistical databases. It allows to efficiently perform machine learning tasks with quantified privacy guarantee while accurately approximating the non-private results.
\begin{definition}{(Differential Privacy \cite{dwork2014algorithmic})} A random algorithm $\mathcal{M}$ is $(\epsilon, \sigma) \text{-}$differentially private if for all $\mathcal{O} \subset Range(\mathcal{M}) $ and for any of all $(\mathcal{D}_0, \mathcal{D}_1 )$ which only differs on one single record such that $||\mathcal{D}_0 - \mathcal{D}_1 || \leq 1$ satisfies
\begin{equation}\label{eq:dp}
Pr[\mathcal{M}(\mathcal{D}_0) \in \mathcal{O}] \leq exp(\epsilon)Pr[(\mathcal{M}(\mathcal{D}_1)\in \mathcal{O}]+\sigma \\ \nonumber
\end{equation}
And $\mathcal{M}$ guarantees $\epsilon \text{-}$differential privacy if $\sigma = 0$.
\end{definition}
The parameter $\epsilon$ states the difference of algorithm $\mathcal{M}$'s output for any $(\mathcal{D}_0, \mathcal{D}_1)$. It measures the privacy loss. Lower $\epsilon$ indicates stronger privacy protection.

\emph{Laplace Mechanism} \cite{dwork2006calibrating} is a common approach to approximate a real-valued function $f: \mathcal{D} \rightarrow \mathbb{R}$ with a differential privacy preservation using additive noise sampled from Laplace distribution:
$\mathcal{M}(\mathcal{D}) \overset{\Delta}{=} f(\mathcal{D}) + Lap(0,\frac{\Delta \mathcal{F}}{\epsilon})$,
where the $\Delta \mathcal{F}$ indicates the largest possible change between the outputs of the function $f$ which takes as input any neighbor databases $(\mathcal{D}_0, \mathcal{D}_1)$. It is referred to as the $L_1$-sensitivity which is defined as:
$\Delta \mathcal{F} = \underset{(\mathcal{D}_0, \mathcal{D}_1)}{max} || f(\mathcal{D}_0) - f(\mathcal{D}_1) ||_1$.

\emph{Sampling} from the posterior distribution of a Bayesian model with bounded log-likelihood, recently, has been proven to be differentially private \cite{wang2015privacy}. It is essentially an \emph{exponential mechanism} \cite{mcsherry2007mechanism}. Formally, suppose we have a dataset of $\mathcal{L}$ i.i.d examples $\mathcal{X} = \{x_i \}^\mathcal{L}_{i=1}$ which we model using a conditional probability distribution $p(x|\theta)$ where $\theta$ is a parameter vector, with a  prior distribution $p(\theta)$. If $p(x|\theta)$ satisfies $sup_{x \in \mathcal{X}, \theta \in \Theta}|\log p(x|\theta)| \leq B $, then releasing one sample from the posterior distribution $p(\theta| \mathcal{X})$ with any prior $p(\theta)$ preserves $4B\text{-}$differential privacy. Alternatively, $\epsilon$ differential privacy can be preserved by simply rescaling the log-posterior distribution by a factor of $\frac{\epsilon}{4B}$, under the regularity conditions where asymptotic normality (Bernstein-von Mises theorem) holds.

\section{Differentially Private SGD}
\label{sec:pgln}

When applying the differential privacy concept, treating the training model (process) as a black box, by only working on the original input or finally output, may result in very poor utility \cite{abadi2016deep,berlioz2015applying}. In contrast, by leveraging the tight characterization of training data, NBM and SGD, we directly calibrate noise into the SGD training process, via Laplace mechanism, to differential-privately learn \emph{similarity}. Algorithm \ref{alg:dpnbmiln} outlines our differentially-private SGD method for training probabilistic NBM.
\begin{algorithm}[h]
   \caption{Differentially Private SGD}
   \label{alg:dpnbmiln}
\begin{algorithmic}[1]
\Require Database $R^{>0}$, privacy parameter $\epsilon$, regular parameter $\lambda$, rescale parameter $\beta$, learning rate $\eta$, the total number of iterations $K$, initialized \emph{similarity} $S^{(1)}$.
\State $S^{(1)} = S^{(1)}\cdot \beta$ \Comment{ rescale the initialization}
\For {$t =1:K$}
\State $\bullet$ uniform-randomly sample a mini-batch $\Phi \subset R^{>0}$.
%\State $e_{max} = 0.5+\frac{\varphi-1}{t+1}$  \Comment{Dynamic error bound}
\State $\Delta \mathcal{F} = 2e_{max}\frac{\tau}{C}$ \Comment{ $e_{max} = 0.5+\frac{\varphi-1}{t+1}$; $|S_i|I_u \geq C$}
%\State $e_{ui} = r_{ui}  - \hat{r}_{ui}$
\State $e_{ui} = min(max(e_{ui},-e_{max}),e_{max})$ \Comment{$e_{ui} = \hat{r}_{ui}-r_{ui}$}
\State $\mathcal{G} = \sum_{(u,i) \in \Phi}e_{ui}\frac{\partial \hat{r}_{ui}}{\partial S_i} + Laplace(\frac{\gamma K \Delta \mathcal{F}}{ \epsilon })$ \Comment{$\gamma = \frac{L }{\mathcal{L}}$}
\State $S^{(t+1)} \leftarrow S^{(t)} - \eta (\beta \mathcal{G} + \lambda S^{(t)})$ \Comment{ up-scale the update}
\EndFor
\State \Return $S^{(t+1)}$
\end{algorithmic}
\end{algorithm}

According to Equation (\ref{myeq:logmpnbm2}) and (\ref{myeq:tgradient}), for each user $u$ (in a randomly chosen mini-batch $\Phi$) the gradient of \emph{similarity} is
\begin{equation}\label{myeq:gdetail}
\begin{split}
\mathcal{G}_{ij}(u) = e_{ui} \frac{\partial \hat{r}_{ui}}{\partial S_{ij}} = e_{ui}(\frac{r_{uj}}{S_iI_u^-}-\hat{r}_{ui}\frac{I_{uj}}{S_iI_u^-})
\end{split}
\end{equation}
where $e_{ui} =\hat{r}_{ui} - r_{ui}$. For the convenience of notation, we omit $S_{ij}<0$ part in Equation (\ref{myeq:gdetail}) which does not compromise the correctness of bound estimation.

%we hereby let $s_{ij} \geq 0$ which does not compromise the correctness of bound estimation.
%\footnote{In the implementation, we allow $s_{ij} < 0$.}

To achieve differential privacy, we update the gradient $\mathcal{G}$ by adding Laplace noise (Algorithm \ref{alg:dpnbmiln}, line 6). The amount of noise is determined by the bound of gradient $\mathcal{G}_{ij}(u)$ (sensitivity $\Delta \mathcal{F}$) which further depends on $e_{ui}, (r_{uj}-\hat{r}_{ui}I_{uj})$ and $|S_i|I_u^-$. We reduce the sensitivity by exploiting the characteristics of training data, NBM and SGD respectively, by the following tricks.

\emph{Preprocessing} is often adopted in machine learning for utility reasons. In our case, it can contribute to privacy protection. For example, we only put users who have more than 20 ratings in the training data. It results in a bigger $|S_i|I_u^-$ thus will reduce sensitivity. Suppose the rating scale is $[r_{min}, r_{max}]$, removing ``paranoid" records makes $|r_{uj}-\hat{r}_{ui}I_{uj}| \leq \varphi$ hold, where $\varphi = r_{max} - r_{min}$.

\begin{figure}[h]
\centering
\includegraphics[height=1.8in, width=3.2in]{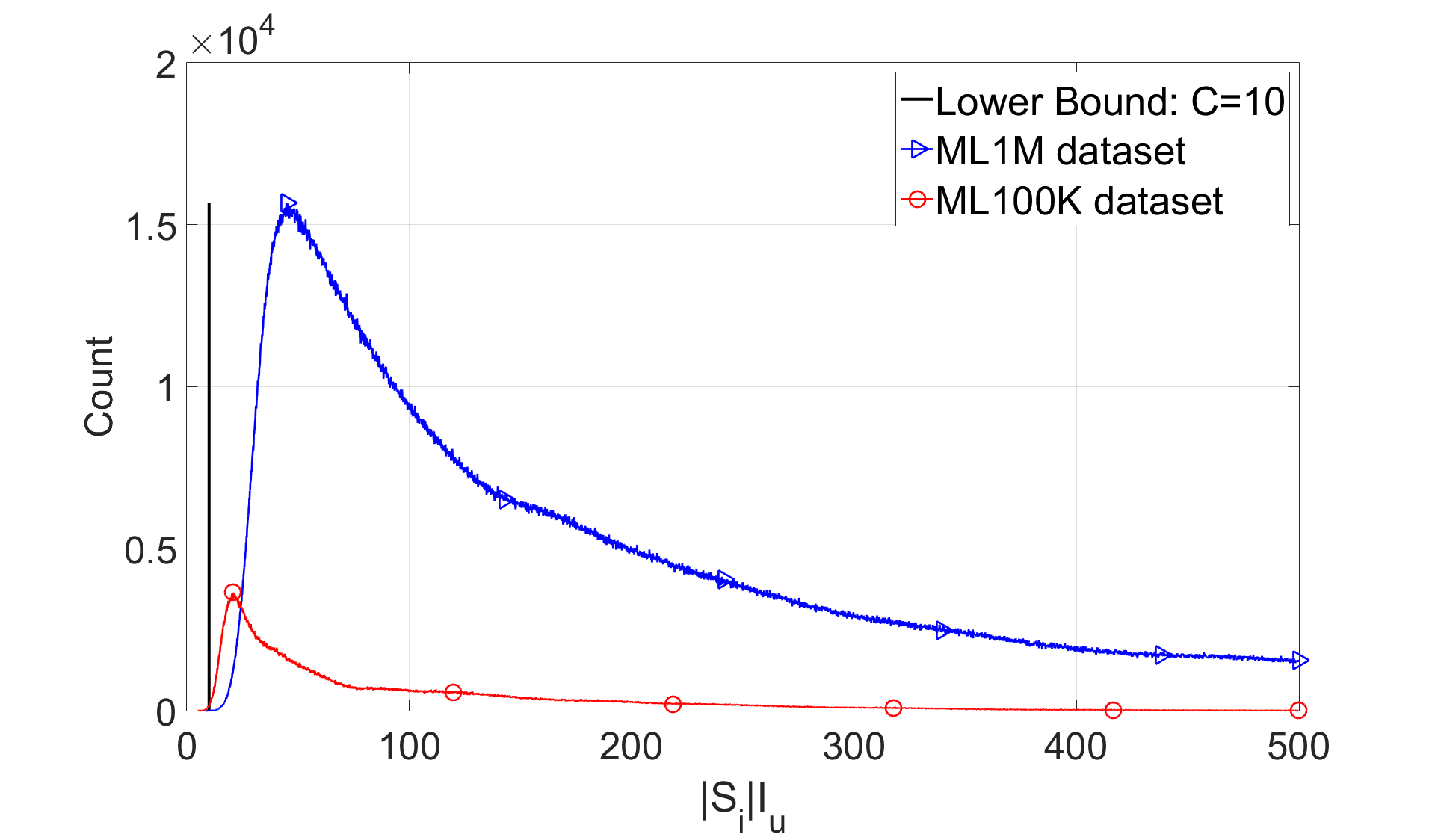}
\caption{ The distribution of $|S_i|I_u$ ($\beta = 10$).  In order to have more detail of the distribution of those points have low $|S_i|I_u$ values, the points $|S_i|I_u \geq 500$ are removed.}
\label{fig:simc}
\end{figure}
%\vspace*{-0.7cm}

\emph{Rescaling the value of similarity} allows a lower sensitivity. NBM, see Equation (\ref{myeq:tnbm}), allows us to rescale the \emph{similarity} $S$ to an arbitrarily large magnitude such that we can further reduce the sensitivity ( by increasing the value of $|S_i|I_u$). However, the initialization of \emph{similarity} strongly influences the convergence of the training. Thus, it is important to balance the convergence (accuracy) and the value of \emph{similarity} (privacy). Another observation is that the gradient down-scales when enlarging the \emph{similarity}, see Equation (\ref{myeq:gdetail}). We can up-scale the gradient monotonically during the training process (Algorithm \ref{alg:dpnbmiln}, line 1 and 7). Fig. \ref{fig:simc} shows , let $\beta=10$, the lower bound of $|S_i|I_u$, denote as $C$, is 10.

\emph{The prediction error} $e_{ui} = \hat{r}_{ui}-r_{ui}$ decreases when the training goes to convergence such that we can clamp $e_{ui}$ to a lower bound dynamically. In our experiments, we bound the prediction error as $|e_{ui}| \leq 0.5 +\frac{\varphi -1 }{t+1}$,
%\begin{equation}
%|e_{ui}| \leq 0.5 +\frac{\varphi -1 }{t+1}
%\end{equation}
where $t$ is the iteration index. This constraint trivially influences the convergence under non-private training process.

After applying all the tricks, we have the dynamic gradient bound at iteration $t$ as follows
\begin{equation}\label{eq:bound}
max(|\mathcal{G}^{(t)}|) \leq (0.5+\frac{\varphi-1}{t+1}) \frac{\varphi}{C}
\end{equation}
The \emph{sensitivity of} each iteration is
$\Delta \mathcal{F} = 2max(|\mathcal{G}^{(t)}|) \leq 2(0.5+\frac{\varphi-1}{t+1}) \frac{\varphi}{C}$.

\begin{theorem} Uniform-randomly sample $L$ examples from a dataset of the size $\mathcal{L}$, Algorithm \ref{alg:dpnbmiln} achieves $\epsilon\text{-}$differential privacy if in each SGD iteration  $t$ we set $\epsilon^{(t)} = \frac{\epsilon}{K\gamma}$ where $K$ is the number of iterations and $\gamma = \frac{L}{\mathcal{L}}$.
\end{theorem}

\begin{proof}
In Algorithm \ref{alg:dpnbmiln}, suppose the number of iterations $K$ is known in advance, and each SGD iteration maintains $\frac{\epsilon}{K\gamma}\text{-}$differential privacy. The privacy enhancing technique \cite{beimel2014bounds,kasiviswanathan2011can} indicates that given a method which is $\epsilon\text{-}$differentially private over a deterministic training set, then it maintains $\gamma \epsilon\text{-}$differential privacy with respect to a full database if we uniform-randomly sample training set from the database where $\gamma$ is the sampling ratio. Finally, combining the privacy enhancing technique with composition theory \cite{dwork2014algorithmic}, it ensures the $K$ iterations SGD process maintain the overall bound of $\epsilon\text{-}$differential privacy. \qed
\end{proof}
%The differentially private item mean-rating is simplly given by
%\begin{equation}\label{eq:avgi}
%\bar{r}_i = \frac{\sum_{u\in R^{>0}}r_{ui}'+Laplace(\frac{\tau}{\epsilon})}{\sum_{u\in R^{>0}}I_{ui}+Laplace(\frac{1}{\epsilon})}
%\end{equation}
%where $r_{ui}' = r_{ui}+\bar{r}_i$.
%Following the composition theory, the total privacy cost of this method is $\epsilon+\epsilon_0$.
 %\vspace*{-.6cm}

\section{Differentially Private Posterior Sampling}
\label{sec:dpps}

Sampling from the posterior distribution of a Bayesian model with bounded log-likelihood has free differential privacy to some extent \cite{wang2015privacy}. Specifically, for probabilistic NBM, releasing a sample of the \emph{similarity} $S$,
\begin{equation}\label{myeq:samples}
\begin{split}
S  \sim p(S|R^{>0}, \alpha_S, \alpha_R)
   \propto exp( \sum_{i=1}^M\sum_{u=1}^N(r_{ui}-\frac{S_iR_u^-}{|S_i|I_u^-})^2+\lambda\sum_{i=1}^M||S_i||_2)
\end{split}
\end{equation}
achieves $4B\text{-}$differential privacy at user level, if each user's log-likelihood is bounded to B, i.e. $\underset{u \in R^{>0}}{max} \sum_{i \in R_u} (\hat{r}_{ui}-r_{ui})^2 \leq B $. Wang et al. \cite{wang2015privacy} showed that we can achieve $\epsilon\text{-}$differential privacy by simply rescaling the log-posterior distribution with $\frac{\epsilon}{4B}$, i.e. $\frac{\epsilon}{4B} \cdot \log p(S|R^{>0}, \alpha_S,\alpha_R)$.

%Following the previous approach \cite{liu2015fast},
Posterior sampling is computationally costly. For the sake of efficiency, we adopt a recent introduced Monte Carlo method, Stochastic Gradient Langevin Dynamics (SGLD) \cite{welling2011bayesian}, as our MCMC sampler. To successfully use SGLD, we need to derive an unbiased estimator of \emph{similarity} gradient from a mini-batch which is a non-trivial task.

Next, we first overview the basic principles of SGLD (Section \ref{sec:sgld}), then we derive an unbiased estimator of the true \emph{similarity} gradient (Section \ref{sec:ssus}), and finally present our privacy-preserving algorithm (Section \ref{sec:subdpps}).

\subsection{Stochastic Gradient Langevin Dynamics}
\label{sec:sgld}

SGLD is an annealing of SGD and Langevin dynamics \cite{rossky1978brownian} which generates samples from a posterior distribution. Intuitively, it adds an amount of Gaussian noise calibrated by the step sizes (learning rate) used in the SGD process, and the step sizes are allowed to go to zero. When it is far away from the basin of convergence, the update is much larger than noise and it acts as a normal SGD process. The update decreases when the sampling approaches to the convergence basin such that the noise dominated, and it behaves like a Brownian motion. SGLD updates the candidate states according to the following rule.
\begin{equation}\label{eq:sgld}
\begin{split}
  \Delta \theta_t  = \frac{\eta_t}{2}(\Delta \log p(\theta_t)+\frac{\mathcal{L}}{L}\sum_{i=1}^{L}\Delta \log p(x_{ti}|\theta_t))+z_t ; \ \ \ \  z_t \sim \mathcal{N}(0,\eta_t)
\end{split}
\end{equation}
where $\eta_t$ is a sequence of step sizes.  $p(x|\theta)$ denotes conditional probability distribution, and $\theta$ is a parameter vector with a prior distribution $p(\theta)$. $L$ is the size of a mini-batch randomly sampled from dataset $\mathcal{X}^\mathcal{L}$.
To ensure convergence to a local optimum, the following requirements of step size $\eta_t$ have to be satisfied:

$$\sum_{t=1}^{\infty}\eta_t = \infty \quad \quad \sum_{t=1}^{\infty}\eta_{t}^2 < \infty$$
%\begin{equation}\label{eq:eta}
%\sum_{t=1}^{\infty}\eta_t = \infty \quad \quad \sum_{t=1}^{\infty}\eta_{t}^2 < \infty
%\end{equation}
Decreasing step size $\eta_t$ reduces the discretization error such that the rejection rate approaches zero, thus we do not need accept-reject test. Following the previous works, e.g. \cite{liu2015fast,welling2011bayesian}, we set step size $\eta_t = \eta_1 t^{-\xi}$, commonly, $\xi \in [0.3, 1]$. In order to speed up the burn-in phase of SGLD, we multiply the step size $\eta_t$ by a temperature parameter $\varrho$ ($0<\varrho < 1$) where $\sqrt{\varrho \cdot \eta_t} \gg \eta_t$ \cite{chen2014stochastic}.

\subsection{Unbiased Estimator of The Gradient }
\label{sec:ssus}
%In our context, the \emph{similarity} is the parameter $\theta$ , and  $R^{>0}$ is the training data (observations).
The log-posterior distribution of \emph{similarity} $S$ has been defined in Equation (\ref{myeq:logmpnbm2}). The true gradient of the \emph{similarity} $S$ over $R^{>0}$ can be computed as
\begin{equation}\label{eq:gofs}
\mathcal{G}(R^{>0}) = \sum_{(u,i)\in R^{>0}}g_{ui}(S;R^{>0}) + \lambda S
\end{equation}
where $g_{ui}(S; R^{>0}) = e_{ui}\frac{\partial \hat{r}_{ui}}{\partial S_i}$.
To use SGLD and make it converge to true posterior distribution, we need an unbiased estimator of the true gradient which can be computed from a mini-batch $\Phi \subset R^{>0}$. Assume that the size of $\Phi$ and $R^{>0}$ are $L$ and $\mathcal{L}$ respectively. The stochastic approximation of the gradient is
\begin{equation}\label{eq:agofs}
\mathcal{G}(\Phi) = \mathcal{L}\bar{g}(S, \Phi) + \lambda S \circ \mathbb{I}[i,j \in \Phi]
\end{equation}
where $\bar{g}(S, \Phi) = \frac{1}{L}\sum_{(u,i)\in \Phi}g_{ui}(S,\Phi)$.  $\mathbb{I} \subset \mathbb{B}^{M \times M}$ is symmetric binary matrix, and  $\mathbb{I}[i,j \in \Phi] =1 $ if any item-pair $(i,j)$ exists in $\Phi$, otherwise 0. $\circ$ presents element-wise product (i.e. Hadamard product). The expectation of $\mathcal{G}(\Phi)$ over all possible mini-batches is,
\begin{equation}\label{eq:egofs}
\begin{split}
\mathbb{E}_{\Phi}  [\mathcal{G}(\Phi)] & =  \mathbb{E}_{\Phi}[\mathcal{L}\bar{g}(S, \Phi)]+ \lambda \mathbb{E}_{\Phi} [ S \circ \mathbb{I}[i,j \in \Phi]] \\
& =  \sum_{(u,i)\in R^{>0}}g_{ui}(S;R^{>0}) + \lambda \mathbb{E}_{\Phi} [ S \circ \mathbb{I}[i,j \in \Phi]]
\end{split}
\end{equation}
$\mathbb{E}_{\Phi} [\mathcal{G}(\Phi)]$ is not an unbiased estimator of the true gradient $\mathcal{G}(R^{>0})$ due to the prior term $\mathbb{E}_{\Phi} [ S \circ \mathbb{I}[i,j \in \Phi]]$. Let $\mathbb{H} = \mathbb{E}_{\Phi} [ \mathbb{I}[i,j \in \Phi]]$, we can remove this bias by multiplying the prior term with $\mathbb{H}^{-1}$ thus to obtain an unbiased estimator. Follow previous approach \cite{ahn2015large}, we assume the mini-batches are sampled with replacement, then $\mathbb{H}$ is,
\begin{equation}\label{eq:hij}
\mathbb{H}_{ij} =  1- \frac{|I_i||I_j|}{\mathcal{L}^2}(1-\frac{|I_j|}{\mathcal{L}})^{L -1}(1-\frac{|I_i|}{\mathcal{L}})^{L -1}
\end{equation}
where $|I_i|$ (resp. $|I_j|$) denotes the number of ratings of item $i$ (resp. $j$) in the complete dataset $R^{>0}$. Then the SGLD update rule is the following:
\begin{equation}\label{eq:upd}
S^{(t+1)} \leftarrow S^{(t)} - \frac{\eta_t}{2}(\mathcal{L}\bar{g}(S^{(t)}, \Phi) + \lambda S^{(t)} \circ \mathbb{H}^{-1})+z_t
\end{equation}

\subsection{Differential Privacy via Posterior Sampling}
\label{sec:subdpps}
%The unbiased estimator of the true \emph{similarity} gradient, which can be computed from a randomly selected subset $\Phi$, enables us to use SGLD.
To construct a differentially private NBM, we exploit a recent observation that sampling from scaled posterior distribution of a Bayesian model with bounded log-likelihood can achieve $\epsilon\text{-}$differential privacy \cite{wang2015privacy}. We summarize the differentially private sampling process (via SGLD) in Algorithm \ref{alg:dpnbmps}.
%\vspace*{-.4cm}
\begin{algorithm}[h]
   \caption{Differentially Private Posterior Sampling (via SGLD)}
   \label{alg:dpnbmps}
\begin{algorithmic}[1]
\Require Temperature parameter $\varrho$, privacy parameter $\epsilon$, regular parameter $\lambda$, initial learning rate $\eta_1$. Let $K$ larger than burn-in phase.
\For {$t = 1:K $}
%\State $\bullet$ Randomly sample a mini-batch $\Phi \subset R^{>0}$ of size $\tau$.
\State  $\bullet$ Randomly sample a mini-batch $\Phi \subset R^{>0}$.
\State $\bar{g}(S^{(t)},\Phi) = \frac{1}{L}\sum_{(u,i) \in \Phi}e_{ui}\frac{\partial \hat{r}_{ui}}{\partial S^{(t)}_i}$ \Comment{gradient of $S$ (mini-batch)}
\State $z_t \sim \mathcal{N}(0, \varrho \cdot \eta_t)$ \Comment{ $\sqrt{\varrho \cdot \eta_t} \gg \eta_t$}
\State $S^{(t+1)} \leftarrow S^{(t)} - \frac{\epsilon}{4B} \cdot \frac{\eta_t}{2}(\mathcal{L}\bar{g}(S^{(t)},\Phi)+\lambda S^{(t)}\circ \mathbb{H}^{-1})+z_t$
\State $\eta_{t+1} = \frac{\eta_1}{t^{\gamma}}$
\EndFor
\State \Return $S^{(t+1)}$
\end{algorithmic}
\end{algorithm}
%\vspace*{-.4cm}

Now, a natural question is how to determine the log-likelihood bound $B$? ($\underset{u \in R^{>0}}{max} \sum_{i \in R_u} (\hat{r}_{ui}-r_{ui})^2 \leq B$, and see Equation (\ref{myeq:samples})). Obviously, $B$ depends on the max rating number per user. To those users who rated more than $\tau$ items, we randomly remove some ratings thus to ensure that each user at most has $\tau$ ratings. In our context, the rating scale is [1,5], let $\tau=200$, we have $B=(5-1)^2 \times 200$ (In reality, most users have less than 200 ratings \cite{liu2015fast}).

\begin{theorem}
\label{the:dpps}
Algorithm \ref{alg:dpnbmps} provides $(\epsilon, (1+e^{\epsilon})\delta)\text{-}$differential privacy guarantee to any user if the distribution $P_{\mathcal{X}}'$ where the approximate samples from is $\delta\text{-}$far away from the true posterior distribution $P_{\mathcal{X}}$, formally $||P_{\mathcal{X}}'-P_{\mathcal{X}} ||_1 \leq \delta $. And $\delta \rightarrow 0$ if the MCMC sampling asymptotically converges.
\end{theorem}

\begin{proof}
Essentially, differential privacy via posterior sampling \cite{wang2015privacy} is an exponential mechanism \cite{mcsherry2007mechanism} which protects $\epsilon\text{-}$differential privacy when releasing a sample $\theta$ with probability proportional to $exp(-\frac{\epsilon}{2\Delta \mathcal{F}}p(\mathcal{X}|\theta))$, where $p(\mathcal{X}|\theta)$ serves as the utility function. If $p(\mathcal{X}|\theta)$ is bounded to $B$, we  have the sensitivity $\Delta \mathcal{F} \leq 2B$. Thus, release a sample by Algorithm \ref{alg:dpnbmps} preserves $\epsilon\text{-}$differential privacy.  It compromises the privacy guarantee to $(\epsilon, (1+e^{\epsilon})\delta)$ if the distribution (where the sample from) is $\delta\text{-}$far away from the true posterior distribution, proved by \cite{wang2015privacy}. \qed
\end{proof}

Note that when $\epsilon = 4B$, the differentially private sampling process is identical to the non-private sampling. This is also the meaning of \emph{some extent of free privacy}. It starts to lose accuracy when $\epsilon < 4B$. One concern of this sampling approach is the distance $\delta$ between the distribution where the samples from and the true posterior distribution, which compromises the differential privacy guarantee. Fortunately, an emerging line of works, such as \cite{sato2014approximation,vollmer2015non}, proved that SGLD can converge in finite iterations. As such we can have arbitrarily small $\delta$ with a (large) number of iterations.
% \cite{MovieLens}

\section{Experiments and Evaluation}
\label{sec:exp}

We test the proposed solutions on two real world datasets, ML100K and ML1M \cite{MovieLens}, which are widely employed for evaluating recommender systems. ML100K dataset has 100K ratings that 943 users assigned to 1682 movies. ML1M dataset contains 1 million ratings that 6040 users gave to 3952 movies. In the experiments, we adopt 5-fold cross validation for training and evaluation. We use root mean square error (RMSE) to measure accuracy performance:
\begin{displaymath}
RMSE = \sqrt{\frac{\sum_{(u,i) \in R^T}(r_{ui}-\hat{r}_{ui})^2}{|R^T|}}
\end{displaymath}
where $|R^T|$ is the total number of ratings in the test set $R^T$. The lower the RMSE value the higher the accuracy. As a result of cross validation, the RMSE value reported in the following figures is the mean value of multiple runs.

\subsection{Experiments Setup}

In the following, the differentially-private SGD based PNBM is referred to as DPSGD-PNBM, and the differentially-private posterior sampling PNBM is referred as DPPS-PNBM. The experiment source code is available at Github\footnote{https://github.com/lux-jwang/Experiments/tree/master/dpnbm}.

We compare their performances with the following (state-of-the-art) baseline algorithms.

\begin{itemize}
\item \emph{{non-private PCC and COS}: } There exist differentially-private NBMs based on Pearson correlation (PCC) or Cosine similarity (COS) NBMs (e.g. \cite{mcsherry2009differentially,zhu2014effective,guerraoui2015d}). Since their accuracy is worse than the non-private algorithms, we directly focus on these non-private ones.
\vspace{0.1cm}
\item \emph{{DPSGD-MF}: } Differentially private matrix factorization from \cite{berlioz2015applying}, which calibrates Laplacian noise into the SGD training process.
\vspace{0.1cm}
\item \emph{{DPPS-MF}: }  Differentially private matrix factorization from \cite{liu2015fast}, which exploits the posterior sampling technique.
\end{itemize}

We empirically choose the optimal parameters for each model using a heuristic grid search method. We summarize them as follows.

\begin{itemize}
\item \emph{{DPSGD-PNBM}: } The learning rate $\eta$ is searched in $\{ 0.1, 0.4 \}$, and the iteration number $K \in [1, 20]$, the regular parameter $\lambda \in \{0.05, 0.005 \}$, the rescale parameter $\beta \in \{10,20 \}$. The neighbor size $N_k = 500$, the lower bound of $|S_i|I_u:\  C \in \{ 10, 15 \}$. In the training process, we decrease $K$ and increase $\{ \eta, C \}$ when requiring a stronger privacy guarantee (a smaller $\epsilon$).
\vspace{0.1cm}
\item \emph{{DPPS-PNBM}: } The initial learning rate $\eta_1 \in \{ 8\cdot 10^{-8}, 4\cdot 10^{-7},  8\cdot 10^{-6}\}$, $\lambda \in \{0.02, 0.002 \}$, the temperature parameter $\varrho = \{0.001, 0.006, 0.09\}$,  the decay parameter $\xi = 0.3$. $N_k = 500$.
\vspace{0.1cm}
\item \emph{{DPSGD-MF}: } $\eta \in \{ 6\cdot 10^{-4}, 8\cdot 10^{-4}\}$, $K \in [10, 50]$ (the smaller privacy loss $\epsilon$ the less iterations), $\lambda \in \{ 0.2, 0.02 \}$, the latent feature dimension $d \in \{10, 15, 20\}$.
\vspace{0.1cm}
\item \emph{{DPPS-MF}: }  $\eta \in \{ 2 \cdot 10^{-9}, 2 \cdot 10^{-8}, 8 \cdot 10^{-7}, 8 \cdot 10^{-6}\}$, $\lambda \in \{ 0.02, 0.05, 0.1, 0.2\}$, $\varrho = \{1\cdot 10^{-4}, 6\cdot 10^{-4}, 4\cdot 10^{-3}, 3\cdot 10^{-2} \}$, $d \in \{10, 15, 20\}$,$\xi = 0.3$.
\vspace{0.1cm}
\item \emph{{non-private PCC and COS}: } For ML100K, we set $N_K=900$. For ML1M, we set $N_K=1300$.
\end{itemize}

\subsection{Comparison Results}
\label{sec:acca}

We first compare the accuracy between DPSGD-PNBM, DPSGD-MF, non-private PCC and COS and show the results in Fig. \ref{fig:dpsgd} for the two datasets respectively.  When $\epsilon \geq 20$, DPSGD-MF does not lose much accuracy, and it is better than non-private PCC and COS. However, the accuracy drops quickly (or, the RMSE increase quickly) when the privacy loss $\epsilon$ is reduced. This matches the observation in \cite{berlioz2015applying}. In the contrast, DPSGD-PNBM maintains a promising accuracy when $\epsilon \geq 1$, and is better than non-private PCC and COS.
\begin{figure}[h!]
\centering
\includegraphics[height=2in, width=5in]{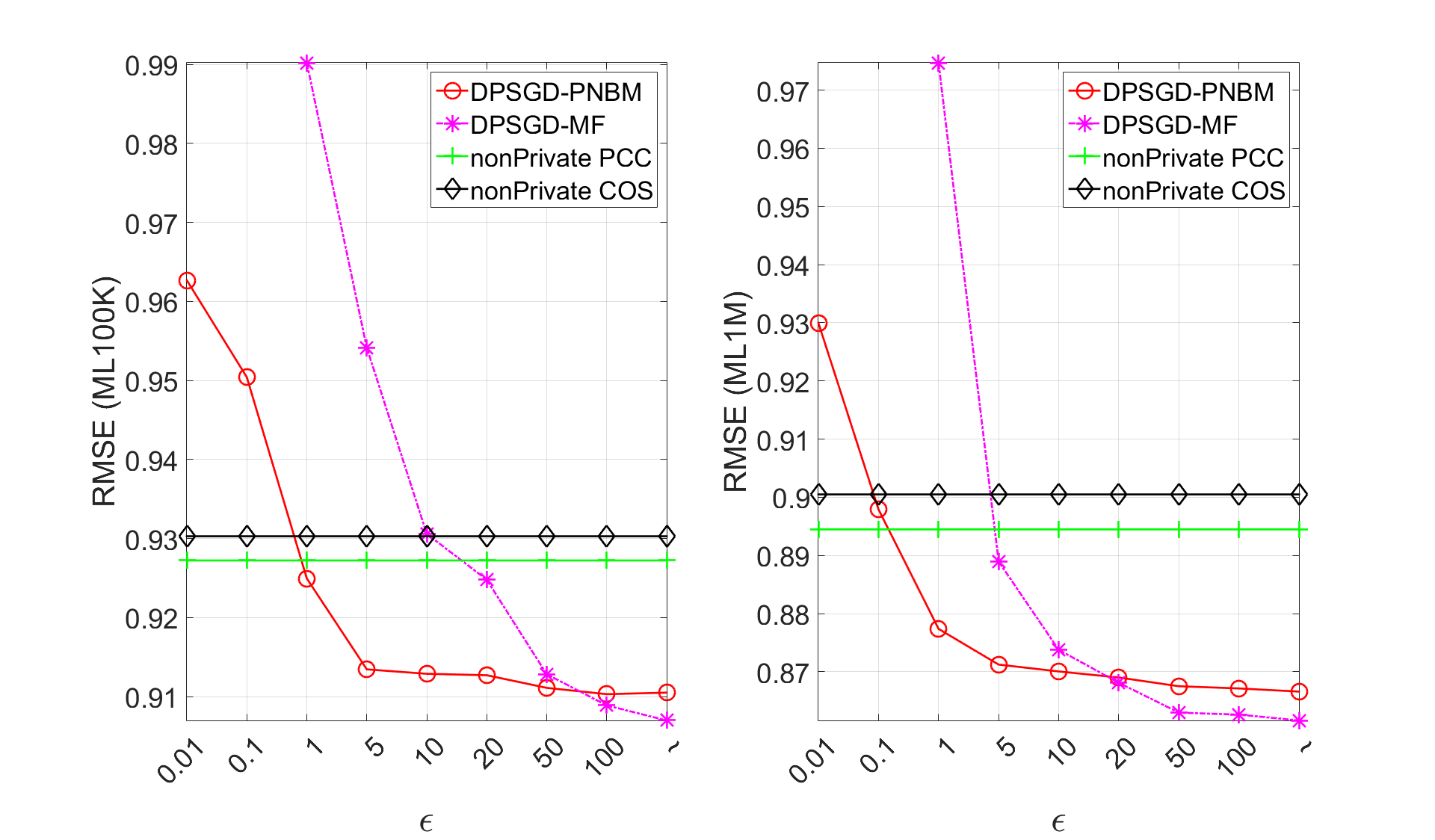}
\caption{Accuracy Comparison: DPSGD-PNBM, DPSGD-MF, non-private PCC, COS. }
\label{fig:dpsgd}
\end{figure}

DPPS-PNBM and DPPS-MF preserve differential privacy at user level. We denote the privacy loss $\epsilon$ in form of $x \times \tau$ where $x$ is a float value which indicates the average privacy loss at a rating level, and $\tau$ is the max rate number per user. The comparison is shown in Fig. \ref{fig:dpps}. In our context, for both datasets, $\tau = 200$. Both DPPS-PNBM and DPPS-MF allow accurate estimations when $\epsilon \geq 0.1 \times 200$. It may seem that $\epsilon=20$ is a meaningless privacy guarantee. We remark that the average privacy of a rating level is 0.1. Besides the accuracy performance is better than the non-private PCC and COS, from the point of privacy loss ratio, our models match previous works \cite{liu2015fast,mcsherry2009differentially}, where the authors showed that differentially private systems may not lose much accuracy when $\epsilon > 1$.
\begin{figure}[h!]
\centering
\includegraphics[height=2in, width=5in]{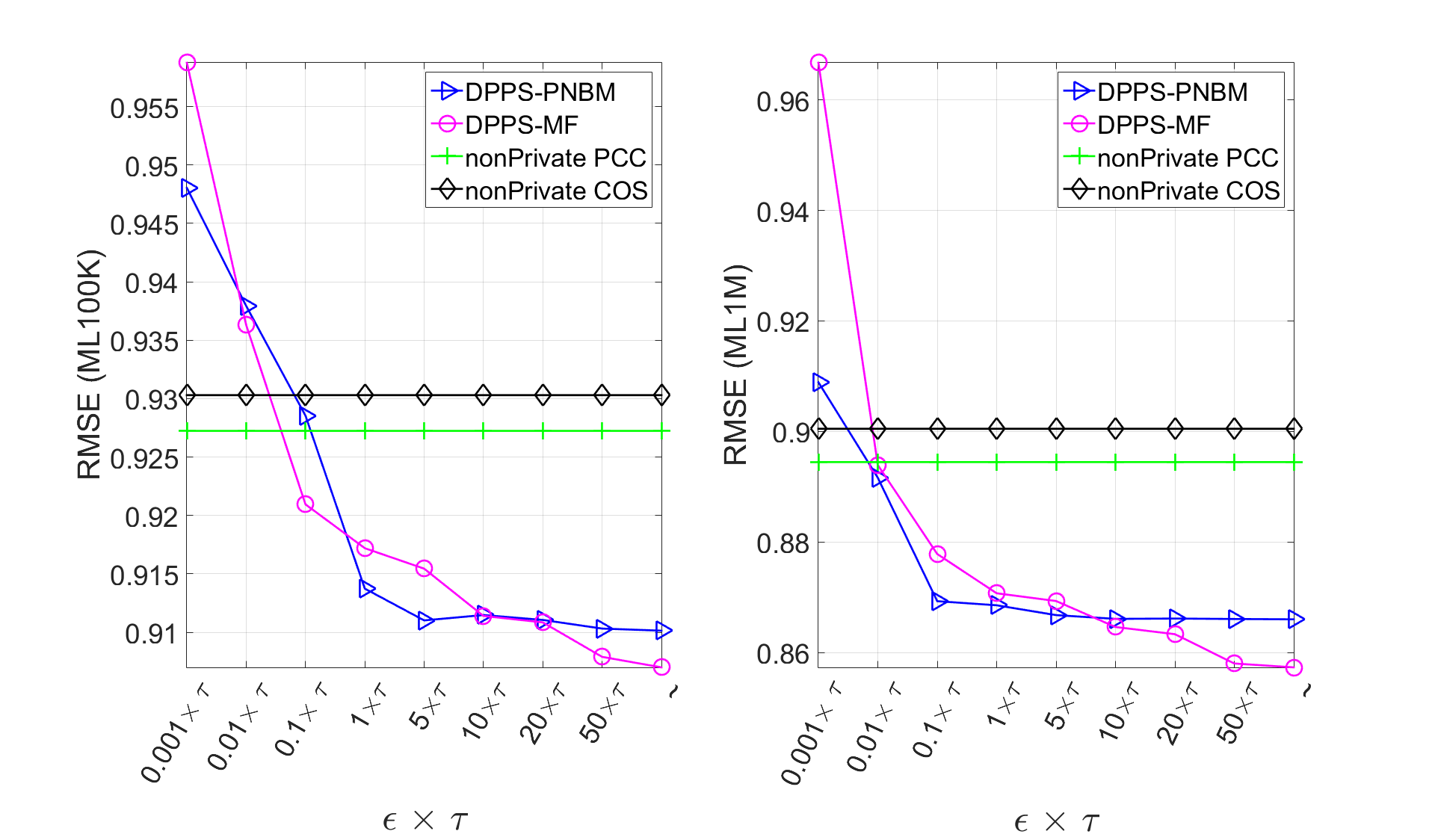}
\caption{Accuracy Comparison: DPPS-PNBM, DPPS-MF, non-private PCC, COS. }
\label{fig:dpps}
\end{figure}

For bandwidth and efficiency reason, mobile service providers may prefer to store the trained model (e.g. item \emph{similarity}) in mobile devices directly. Commercial recommender systems often have very large \emph{similarity} matrix such that the shortage of memory space in mobile devices may become a bottleneck. In order to alleviate this issue, we choose the $Top \text{-}N$ most similar neighbors only by \emph{similarity} matrix, by removing the rest neighbors of each item, such that we can sparsely store the matrix in practice. We compare accuracy with different number of neighbors with $\epsilon=1$, and summarize the results in Fig. \ref{fig:nk}. We stress two observations. Both DPSGD-PNBM and DPPS-PNBM reach their best accuracy with a smaller neighbor size. The accuracy of both DPSGD-PNBM and DPPS-PNBM is less sensitive than PCC and COS, when neighbor size is changed. This helps mitigate over-fitting problem and enhance system robustness.
\begin{figure}[h!]
\centering
\includegraphics[height=2in, width=5in]{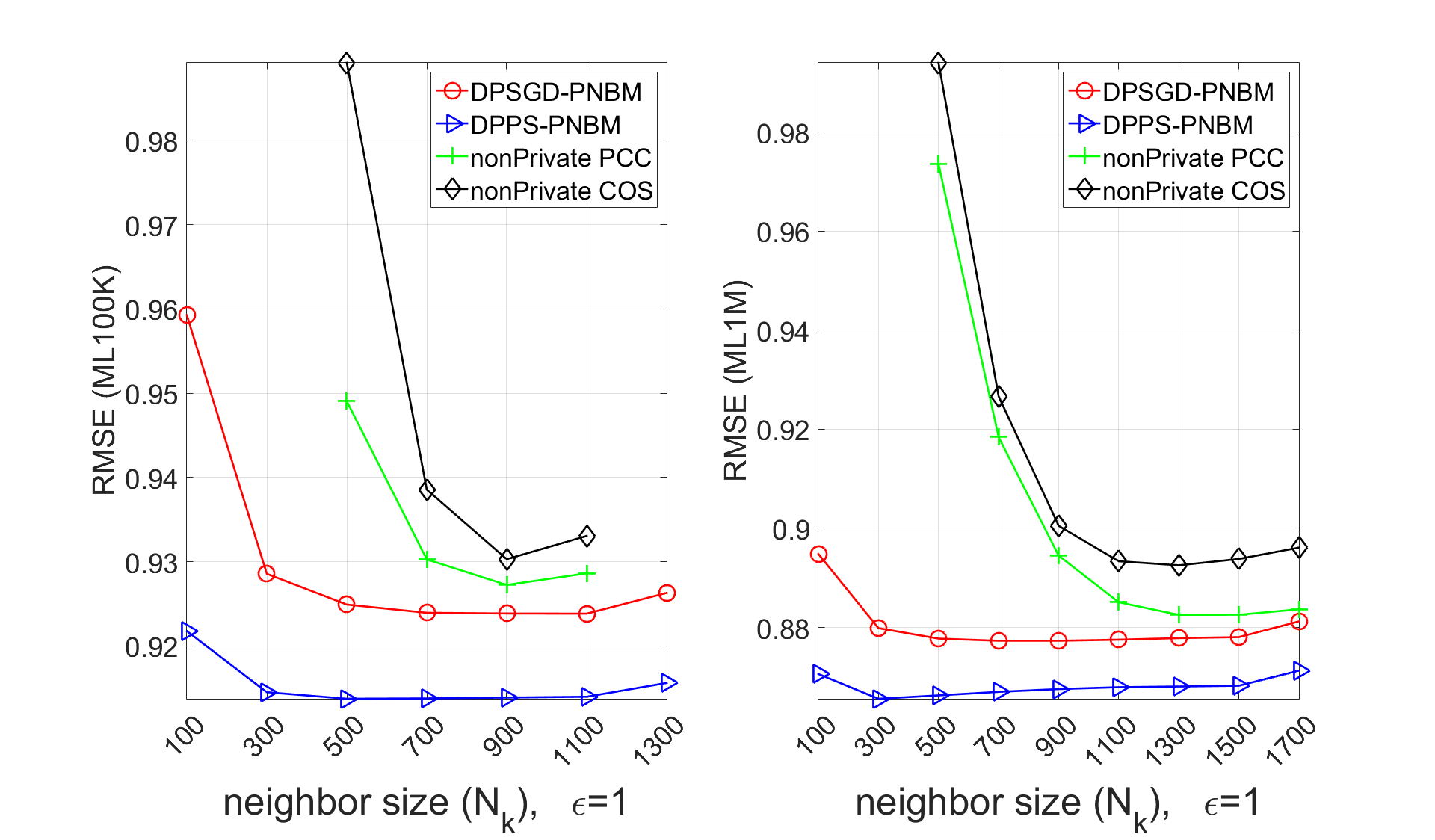}
\caption{Accuracy comparison with different neighbor sizes}
\label{fig:nk}
\end{figure}

DPSGD-PNBM and DPPS-PNBM achieve differential privacy at rating level (a single rating) and user level (a whole user profile) respectively. Below, we try to compare them at rating level, precisely at the average rating level for DPPS-PNBM. Fig. \ref{fig:sgdps} shows that both solutions can obtain quite accurate predictions with a privacy guarantee ($\epsilon \approx  1$). With the same privacy guarantee, DPPS-PNBM seems to be more accurate. However, DPPS-PNBM has its potential drawback. Recall from Section \ref{sec:dpps}, the difference $\delta$ between the distribution where samples from and the true posterior distribution compromises differential privacy guarantee. In order to have an arbitrarily small $\delta$, DPPS-PNBM requires a large number of iterations \cite{sato2014approximation,vollmer2015non}. At this point, it is less efficient than DPSGD-PNBM. In our comparison, we assume $\delta \rightarrow 0$.

\begin{figure}[h]
\centering
\includegraphics[height=2in, width=4.2in]{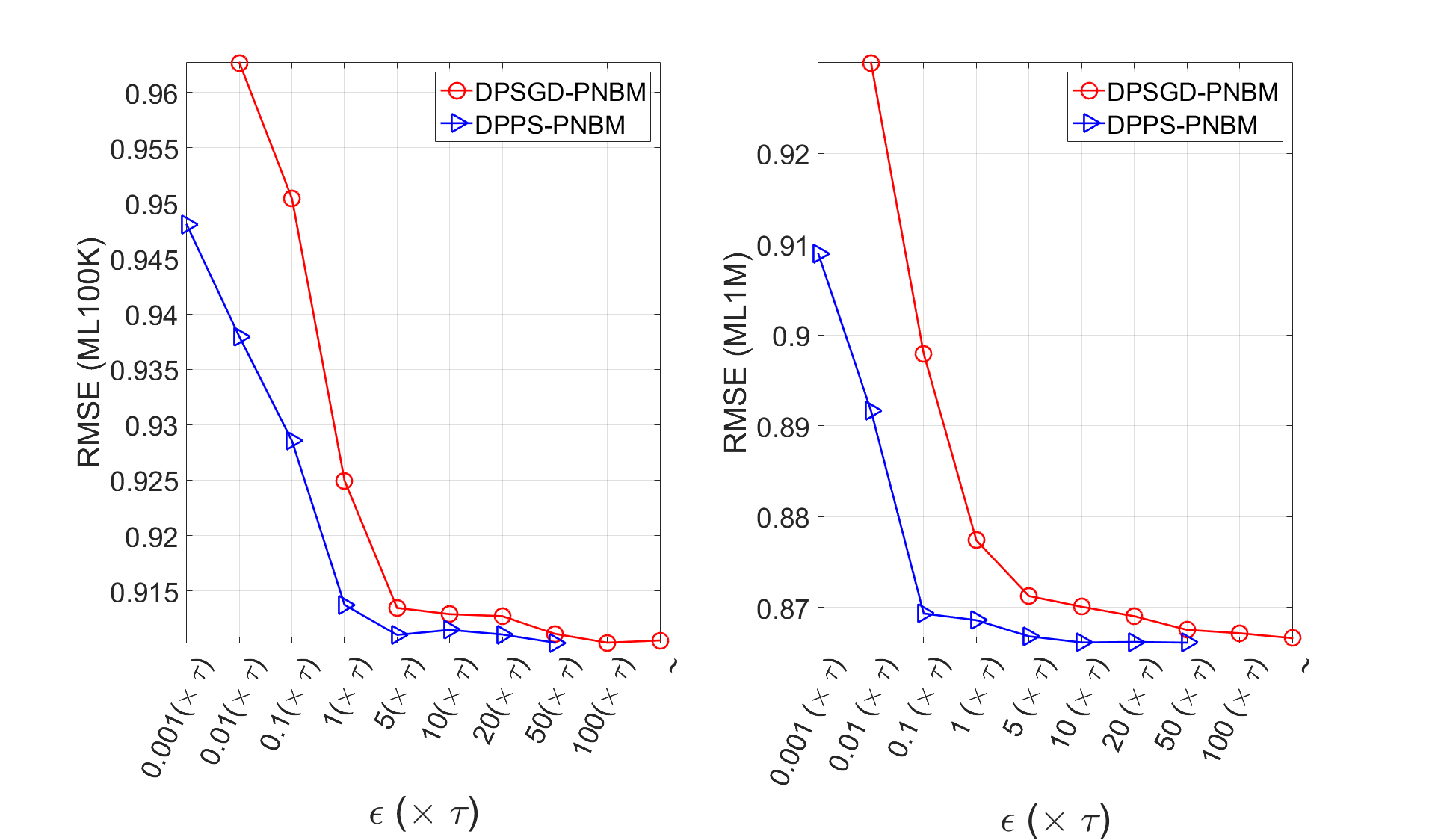}
\caption{Accuracy comparison between DPSGD-PNBM and DPPS-PNBM}
\label{fig:sgdps}
\end{figure}

\subsection{Summary}

In summary, DPSGD-MF and DPPS-MF are more accurate when privacy loss is large (e.g. in a non-private case). DPSGD-PNBM and DPPS-PNBM are better when we want to reduce the privacy loss to a meaningful range. Both our models consistently outperform non-private traditional NBMs, with a meaningful differential privacy guarantee. Note that \emph{similarity} is independent of NBM itself, thus other neighborhood-based recommenders can use our models to differential-privately learn \emph{Similarity}, and deploy it to their existing systems without requiring extra effort.

\section{Related Work}
\label{sec:related}

%The privacy risk inherent in collaborative filtering has brought great concern to public.
A number of works have demonstrated that an attacker can infer the user sensitive information, such as gender and politic view, from public recommendation results without using much background knowledge \cite{calandrino2011you,friedman2015privacy,narayanan2008robust,weinsberg2012blurme}.

Randomized data perturbation is one of earliest approaches to prevent user data from inference attack in which people either add random noise to their profiles or substitute some randomly chosen ratings with real ones (e.g.\cite{polat2003privacy,polat2005privacy,polat2006achieving}). While this approach is very simple, it does not offer rigorous privacy guarantee. Differential privacy \cite{dwork2014algorithmic} aims to precisely protect user privacy in statistical databases, and the concept has become very popular recently. \cite{mcsherry2009differentially} is the first work to apply differential privacy to recommender systems, and it has considered both neighborhood-based methods (using correlation as \emph{similarity}) and latent factor model (e.g. SVD). \cite{zhu2014effective} introduced a differentially private neighbor selection scheme by injecting Laplace noise to the \emph{similarity} matrix. \cite{guerraoui2015d} presented a scheme to obfuscate user profiles that preserves differential privacy. \cite{berlioz2015applying,liu2015fast} applied differential privacy to matrix factorization, and we have compared our solutions to theirs in Section \ref{sec:exp}.

Secure multiparty computation (SMC) recommender systems allow users to compute recommendation results without revealing their inputs to other parties. Many protocols have been proposed in the literature, e.g.  \cite{canny2002collaborative,tang2015privacy,nikolaenko2013privacy}. Unfortunately, these protocols do not prevent information leakage from the recommendation results.

\section{Conclusion}
\label{sec:con}

In this paper, we have proposed two different differentially private NBMs,  under a probabilistic framework. We firstly introduced a way to differential-privately find the maximum a posteriori \emph{similarity} by calibrating noise to the SGD training process. Then we built differentially private NBM by exploiting the fact that sampling from scaled posterior distribution can result in differentially private systems. While the experiment results have demonstrated that our models allow promising accuracy with a modest privacy budget in some well-known datasets, we consider it as an interesting future work to test the performances in other real world datasets.

\section*{Acknowledgments}
Both authors are supported by a CORE (junior track) grant from the National Research Fund, Luxembourg.

\bibliographystyle{abbrv}
\bibliography{sigproc}
\end{document}